\documentclass[11pt]{article}
\usepackage{amsmath,amssymb,amsthm}
\usepackage{fullpage}
\usepackage{microtype}
\usepackage{hyperref}
\usepackage{graphicx}
\usepackage{caption}  
\DeclareMathOperator{\maxrest}{maxrest}
\DeclareMathOperator{\maxneed}{maxneed}
\title{Heaven \& Hell II: Scale Laws and Robustness in One-Step Heaven-Hell Consensus\thanks{GUDraBIT Research Technical Report TR-2025-02.}}
\author{
  Nnamdi Aghanya \\ \texttt{nnamdi.aghanya@cranfield.ac.uk}
  \and
  Romain Leemans \\ \texttt{romain.leemans@protonmail.com}
}

\date{\today}
\newtheorem{theorem}{Theorem}
\newtheorem{lemma}[theorem]{Lemma}
\newtheorem{corollary}[theorem]{Corollary}

\theoremstyle{definition}

\begin{document}
\maketitle
\begin{abstract}
\noindent
\noindent
We study the Heaven-Hell dynamics, a model for network consensus. A known result establishes an exact one-step convergence threshold for systems with a single uniform hub: the per-node inbound hub weight $W$ suffices if and only if $W \ge \maxrest$, where $\maxrest$ is the maximum non-hub inbound mass. In this work, we develop \emph{scale laws} and \emph{operational refinements} that make this threshold robust to tie-breaking policies, node-specific tolerances, targeted seeding, multiple hubs, and asynchronous update schedules. Our main contributions are: (i) a conservation-law perspective that yields a majority-form characterisation and monotonicity principles; (ii) a parameterised tie policy with exact characterisations for \textsc{TieGlory}/\textsc{TieGnash}/\textsc{TieStay}; (iii) tighter \emph{pointwise deg–max} bounds that strictly improve the classical $\mathrm{indeg}\cdot\mathrm{wmax}$ worst-case guarantee; (iv) a one-pass fairness result for asynchronous updates; and (v) practically checkable sufficient conditions for seeded one- and two-step convergence, including multi-hub budget splits. All proofs are mechanised in Coq, and experiments on rings, grids, BA/scale-free graphs, and heterogeneous weighted graphs validate the tightness and show substantial gap closures over prior bounds~\footnote{\href{https://github.com/gudfit/HeavenHellII}{https://github.com/gudfit/HeavenHellII}}..
\end{abstract}

\section{Introduction}
The Heaven-Hell (HH) model describes consensus dynamics on weighted directed graphs, which is a binary State Consensus on weighted directed graphs. In the HH model, each node updates its state based upon a weighted majority vote of all neighbours; however, there exists a single designated hub node that will enforce a Glory state and cause convergence of all other nodes. The paper builds on previously established work that described exact one-step convergence thresholds for uniform-hub systems\cite{AghanyaLeemans2025}. 

This paper establishes scale laws that increase robustness against realistic variations in the system, such as node-specific tolerances, parameterised tie-breaking policies, multi-hub seeding strategies and asynchronous update schedules. Using a conservation-law perspective, we derived majority form characterisation and monotonicity principles that provide tighter pointwise bounds, exact thresholds, and conditions for both seeded and multi-step convergence. We provide empirical validations of the results using diverse families of graphs, including rings, grids, scale-free networks, and heterogeneous weighted graphs. These results confirmed the tightness of our bounds and demonstrated significant improvements over classical worst-case guarantees.

\section{Preliminaries: The Heaven-Hell Model}
We work on a finite vertex set $V$ with a distinguished hub $g\in V$. Let $w:V\times V\to\mathbb{N}$ be nonnegative weights representing directed edges, and let $\tau:V\to\mathbb{N}$ be an optional tolerance per node. A configuration (state) is a function $s:V\to\{\textsc{Glory},\textsc{Gnash}\}$. We define the inbound weight to a node $v$ from the hub as $\mathrm{hub\_weight}(v)=\sum_{u=g} w(u,v)$ and from all other nodes as \emph{rest weight}, $\mathrm{rest\_weight}(v)=\sum_{u\ne g} w(u,v)$. In the uniform-hub setting, we assume $\mathrm{hub\_weight}(v)=W$ for all $v\ne g$.

A single update step first forces the hub $g$ to \textsc{Glory}. Then, each node $v\ne g$ adopts the \textsc{Glory} state if and only if its score from the \textsc{Glory}-state neighbours dominates its score from the \textsc{Gnash}-state neighbours, with the handling of ties dependent on the active policy.

\paragraph{Notation.} We denote by $\mathrm{score}_G(s,v)$ and $\mathrm{score}_N(s,v)$ the weighted inbound score to $v$ from the \textsc{Glory} and \textsc{Gnash} nodes, respectively, after forcing $g$ to \textsc{Glory}. We write $\mathrm{total\_in}(v)=\sum_u w(u,v)$. The one-step next state under tolerance $\tau$ and a \textsc{TieGlory} policy is denoted $\mathrm{next\_heaven\_tau}(s,v)$; see \S\ref{sec:ties} for details on policies.

\section{A Conservation Law and a Majority-Form View}
A simple but powerful identity underpins our analysis. The total inbound weight to a node is conserved, perfectly partitioning between the scores from nodes in \textsc{Glory} and \textsc{Gnash}.

\begin{lemma}[Score Conservation]
\label{lem:score_conservation}
For any state $s$ and node $v$, the scores after forcing the hub are related by:
\[
 \mathrm{score}_G(s,v) + \mathrm{score}_N(s,v) = \mathrm{total\_in}(v).
\]
\end{lemma}
\begin{proof}
By definition, $\mathrm{total\_in}(v) = \sum_{u \in V} w(u,v)$. The scores are given by weighted sums over indicators: $\mathrm{score}_G(s,v) = \sum_{u \in V} w(u,v) \cdot \mathbb{I}[s'(u) = \textsc{Glory}]$ and $\mathrm{score}_N(s,v) = \sum_{u \in V} w(u,v) \cdot \mathbb{I}[s'(u) = \textsc{Gnash}]$, where $s'$ is the state after forcing the hub. Since any node is either in \textsc{Glory} or \textsc{Gnash}, the indicators satisfy $\mathbb{I}[s'(u) = \textsc{Glory}] + \mathbb{I}[s'(u) = \textsc{Gnash}] = 1$ for all $u$. The result follows from distributing the sum.
\end{proof}

This conservation law allows us to re-characterise the update rule in a majority form.

\begin{theorem}[Majority Form with Tolerance]
\label{thm:majority-tau}
For any node $v\ne g$, its next state is \textsc{Glory} if and only if:
\[
 \mathrm{next\_heaven\_tau}(s,v)=\textsc{Glory}
 \quad\Longleftrightarrow\quad
 2\cdot \mathrm{score}_G(s,v) + \tau(v) \ge \mathrm{total\_in}(v).
\]
\end{theorem}
\begin{proof}
For a non-hub node $v$, the \textsc{TieGlory} update rule is $\mathrm{next\_heaven\_tau}(s,v) = \textsc{Glory} \iff \mathrm{score}_N(s,v) \le \mathrm{score}_G(s,v) + \tau(v)$. We substitute $\mathrm{score}_N(s,v)$ using Lemma \ref{lem:score_conservation}: $\mathrm{total\_in}(v) - \mathrm{score}_G(s,v) \le \mathrm{score}_G(s,v) + \tau(v)$. Rearranging the terms yields the equivalent condition $\mathrm{total\_in}(v) \le 2 \cdot \mathrm{score}_G(s,v) + \tau(v)$.
\end{proof}

\section{Monotonicity Principles}
The majority-form characterisation makes monotonicity properties immediate.

\paragraph{State Monotonicity.} If a state $t$ has weakly more \textsc{Glory} nodes than a state $s$, then for any node $v$, its \textsc{Glory}-score weakly increases and its \textsc{Gnash}-score weakly decreases. Consequently, any node that transitions to \textsc{Glory} from state $s$ will also do so from state $t$. The formal proof is deferred to Lemma \ref{lem:state_monotonicity_appendix}.

\paragraph{Tolerance Monotonicity.} If a tolerance function $\tau'$ pointwise dominates another, $\tau'(v)\ge \tau(v)$ for all $v$, then the required hub weight for convergence can only decrease: $\maxneed(\tau') \le \maxneed(\tau)$, where $\maxneed(\tau)=\max_{v\ne g}\big(\mathrm{rest\_weight}(v)-\tau(v)\big)$. Thus, increasing tolerances helps achieve consensus.

\section{Tie Policies and Exact Characterisations}
\label{sec:ties}
We parameterise the tie-breaking policy by three rules: \textsc{TieGlory}, \textsc{TieGnash}, and \textsc{TieStay}. Let $T(v)=\mathrm{score}_G(s,v)+\tau(v)$ and $R(v)=\mathrm{score}_N(s,v)$ be the total \textsc{Glory} and \textsc{Gnash} scores for a node $v \neq g$ after forcing the hub.

\begin{lemma}[Tie Policies]
The next state of a node $v \neq g$ is determined as follows:
\begin{itemize}
\item \textsc{TieGlory}: $\mathrm{next}=\textsc{Glory}$ iff $R(v) \le T(v)$.
\item \textsc{TieGnash}: $\mathrm{next}=\textsc{Glory}$ iff $R(v) < T(v)$.
\item \textsc{TieStay}: $\mathrm{next}=\textsc{Gnash}$ if $T(v) < R(v)$, \textsc{Glory} if $R(v) < T(v)$, otherwise stay $s(v)$.
\end{itemize}
\end{lemma}
These distinctions are critical in knife-edge cases and when stability (\textsc{TieStay}) is preferred over a bias toward \textsc{Glory}.

\section{Scale Laws: Exact Thresholds and Tighter Bounds}
We now establish the main threshold results for one-step convergence.

\begin{theorem}[Uniform Hub with Tolerance (Exact)]
\label{thm:uniform-tau}
Assume $\mathrm{hub\_weight}(v)=W$ for all $v\ne g$. One-step convergence to all-\textsc{Glory} from any initial state is guaranteed if and only if:
\[
 W \ge \maxneed(\tau) := \max_{v\ne g}\big(\mathrm{rest\_weight}(v) - \tau(v)\big).
\]
\end{theorem}
\begin{proof}
The system converges in one step if and only if for every state $s$ and every node $v \neq g$, $\mathrm{next\_heaven\_tau}(s,v) = \textsc{Glory}$. This is equivalent to requiring that the condition holds even in the worst-case initial state, where all non-hub nodes start in \textsc{Gnash}. In this state, the \textsc{Glory} score for a node $v$ is exactly its hub weight, $\mathrm{score}_G(s,v) = \mathrm{hub\_weight}(v) = W$, and its \textsc{Gnash} score is its rest weight, $\mathrm{score}_N(s,v) = \mathrm{rest\_weight}(v)$. The update rule $\mathrm{score}_N \le \mathrm{score}_G + \tau(v)$ becomes $\mathrm{rest\_weight}(v) \le W + \tau(v)$ for all $v \neq g$. This per-node condition is equivalent to the aggregated threshold $W \ge \max_{v \neq g}(\mathrm{rest\_weight}(v) - \tau(v))$, as formally shown in Lemma \ref{lem:max_need_equiv_appendix}. State monotonicity ensures that if convergence occurs from the worst-case state, it also occurs from any other state.
\end{proof}

\begin{corollary}[Uniform Hub, No Tolerance]
\label{thm:uniform-notau}
Setting $\tau(v)=0$ for all $v$ in Theorem \ref{thm:uniform-tau}, one-step all-\textsc{Glory} is guaranteed if and only if $W \ge \max_{v\ne g}\mathrm{rest\_weight}(v)$.
\end{corollary}

\paragraph{Pointwise Deg–Max Bound.} Let $\mathrm{indeg}_{\neg g}(v)$ be the number of non-hub in-neighbours of $v$, and $\mathrm{max\_in}_{\neg g}(v)$ be the maximum inbound edge weight to $v$ from a non-hub. We have $\mathrm{rest\_weight}(v) \le \mathrm{indeg}_{\neg g}(v) \cdot \mathrm{max\_in}_{\neg g}(v)$. This gives a strictly tighter worst-case bound than the classical global product:
\[
 \maxrest \le \max_{v\ne g} \Big(\mathrm{indeg}_{\neg g}(v)\cdot \mathrm{max\_in}_{\neg g}(v)\Big) \le \mathrm{indeg}_{\text{global}}\cdot \mathrm{wmax}_{\text{global}}.
\]

\section{Seeding and Multi-Hub Variants}
Let $H\subseteq V$ be a seed set of nodes initially forced to \textsc{Glory}. Define your collective influence as $\mathrm{hub\_H}(v) := \sum_{u\in H} w(u,v)$ and your opposition as $\mathrm{rest\_H}(v) := \sum_{u\notin H} w(u,v)$.

\begin{theorem}[Exact One-Step Criterion under Seeded Forcing]
\label{thm:seed-exact}
For any seed set $H$, one-step convergence to all-\textsc{Glory} is guaranteed if and only if for all $v \notin H$:
\[
 \mathrm{hub\_H}(v)+\tau(v)\ \ge\ \mathrm{rest\_H}(v).
\]
\end{theorem}
\begin{proof}
\begin{itemize}
    \item[($\Leftarrow$)] Assume that the condition holds. For any node $v \notin H$ and any state $s$, its score from \textsc{Glory} nodes after forcing $H$ is $\mathrm{score}_G(s,v) \ge \mathrm{hub\_H}(v)$, while its score from \textsc{Gnash} is $\mathrm{score}_N(s,v) \le \mathrm{rest\_H}(v)$. The update condition $\mathrm{score}_N \le \mathrm{score}_G + \tau(v)$ is satisfied because $\mathrm{rest\_H}(v) \le \mathrm{hub\_H}(v) + \tau(v) \le \mathrm{score}_G(s,v) + \tau(v)$.
    \item[($\Rightarrow$)]  Assume that one-step convergence holds. Consider the initial state where all nodes $u \notin H$ are \textsc{Gnash}. For any $v \notin H$, the scores are exactly $\mathrm{score}_G(s,v) = \mathrm{hub\_H}(v)$ and $\mathrm{score}_N(s,v) = \mathrm{rest\_H}(v)$ (see Lemma \ref{lem:all_gnash_scores_appendix}). Since $v$ must become \textsc{Glory}, the update condition $\mathrm{score}_N \le \mathrm{score}_G + \tau(v)$ implies $\mathrm{rest\_H}(v) \le \mathrm{hub\_H}(v) + \tau(v)$.
\end{itemize}
\end{proof}

\section{Asynchronous Schedules: One-Pass Fairness}
Let $\mathrm{run\_sched\_tau}(s,\mathrm{sched})$ denote a single sweep in which the nodes are updated sequentially.

\begin{theorem}[One-Pass Fairness]
\label{thm:async-one-pass}
If the one-step domination condition $\mathrm{hub\_weight}(v)+\tau(v)\ge \mathrm{rest\_weight}(v)$ holds for all $v\ne g$, and a finite schedule visits every non-hub node at least once, then after one pass every $v\ne g$ will be in the \textsc{Glory} state, regardless of the starting configuration.
\end{theorem}
\begin{proof}
The domination condition implies that for any state $s$, $\mathrm{next\_heaven\_tau}(s,v) = \textsc{Glory}$ for all $v \neq g$. When a node $v$ is updated according to the schedule, its state becomes \textsc{Glory}. By state monotonicity, once a node becomes \textsc{Glory}, it cannot revert to \textsc{Gnash} in subsequent updates within the same pass, as the set of \textsc{Glory} nodes only grows. Since the schedule covers all non-hub nodes, each will be updated to \textsc{Glory} at its turn and will remain so.
\end{proof}

\section{Experiments}
\label{sec:experiments}
We summarise the experiments that validate the theoretical results; see Figure~\ref{fig:growth-summary}.

\paragraph{Regular Families.} Rings with $k$-nearest neighbours: The exact threshold is $W^*=2k$. 2D grid: With degree $d=4$, the theoretical threshold $W^*=4$ is matched experimentally (Fig~\ref{fig:growth-summary}A).

\paragraph{Heterogeneous Weighted Graphs.} On an adversarially constructed graph, our pointwise deg–max bound is exact, while the classical bound is off by a factor of 200 (Fig~\ref{fig:growth-summary}B):
\[
 \maxrest = 800,\qquad
 \text{pointwise deg--max} = 800\ (\text{exact}),\qquad
 \text{classical} = 160{,}000\ (\times 200).
\]

\paragraph{Seeding and Multi-Hubs.} On a ring with $W^*=6$, a single hub at $W=5$ fails. However, splitting a budget of $W=6$ between two hubs ($3+3$) or three hubs ($2+2+2$) and seeding them achieves a one-step all-\textsc{Glory}, consistent with Theorem \ref{thm:seed-exact} (Fig~\ref{fig:growth-summary}C).

\paragraph{Asynchronous Fairness.} On a ring and a grid satisfying the one-step condition, schedules that cover all non-hubs in one pass yield all-\textsc{Glory} across random orderings, matching Theorem \ref{thm:async-one-pass} (Fig~\ref{fig:growth-summary}D).

\begin{figure}[t]
\centering
\includegraphics[width=\linewidth]{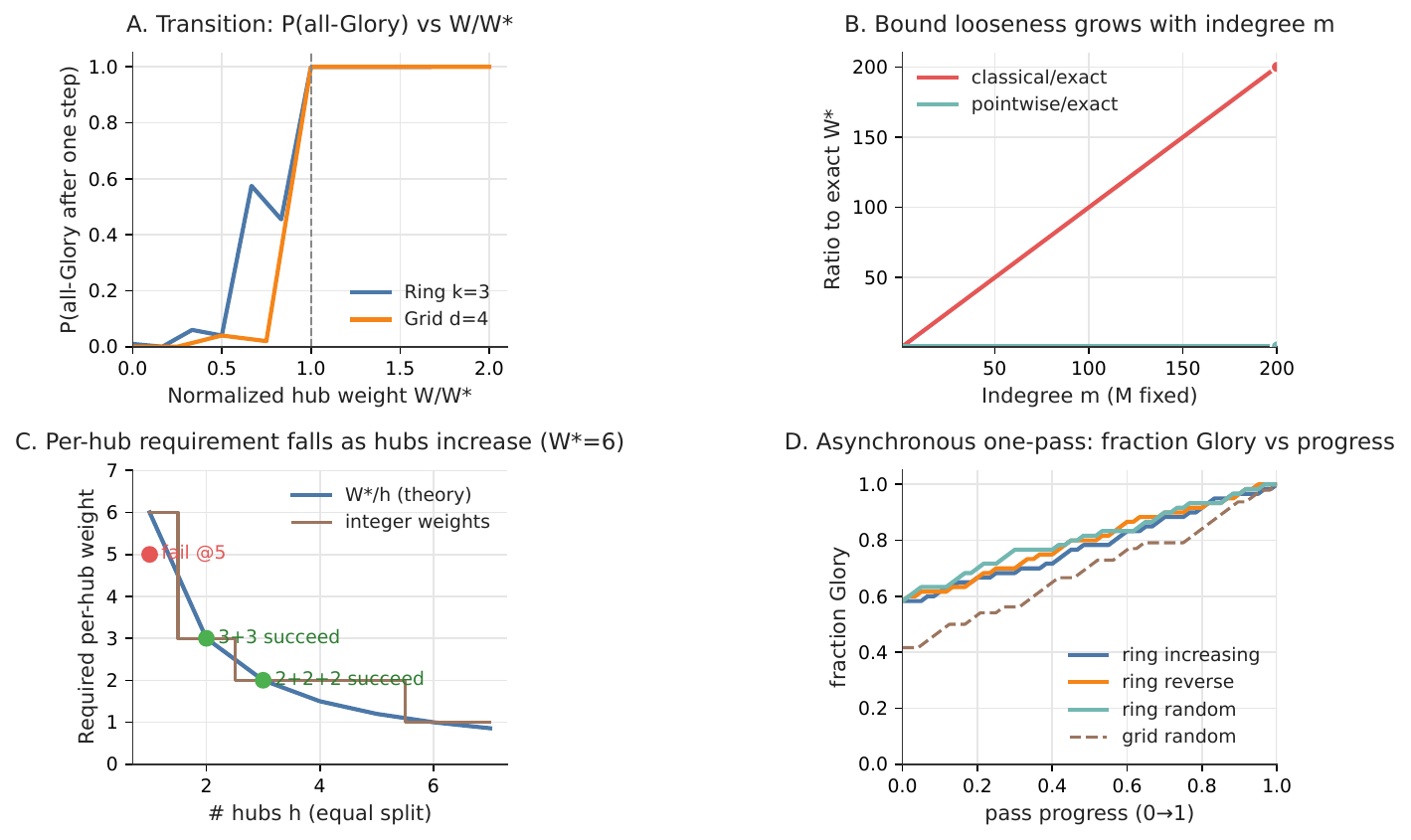}
\caption{Dynamics and scale effects. (A) One-step success probability rises sharply at $W/W^*=1$ on rings and grids. (B) On heterogeneous graphs, the classical deg$\cdot$wmax bound diverges with indegree $m$ while the pointwise bound stays tight. (C) Equal multi-hub splits reduce per-hub weight as $1/h$; annotated examples match the text. (D) Under domination, asynchronous one-pass schedules monotonically drive the fraction of Glory to 1.}
\label{fig:growth-summary}
\end{figure}

\section{Related Work and Limitations}

The Heaven-Hell model aligns with the broader literature on threshold-based dynamics in networks, including weighted majority dynamics \cite{Peleg1998} and bootstrap percolation \cite{Chalupa1979, AdlerLev1991}. Our conservation law and majority-form views extend monotone system analyses \cite{Smith1995}, providing exact per-node thresholds that contrast with probabilistic bounds in influence maximisation \cite{Kempe2003}. Unlike heuristics in complex contagion \cite{CentolaMacy2007}, our analysis provides exact conditions for deterministic one-step convergence, bridging theoretical guarantees with practical refinements absent in prior hub-based consensus models \cite{AghanyaLeemans2025}.

Our one-step guarantees, however, have their bounds and limitations. In particular, one open and important area for further archaeological work is extending these exact results to optimal multi-step policies, or extensible-time varying weights or stochastic observations. The per-node thresholds, although exact in the case of our hub-forcing semantics, would be overly conservative if the hub itself were forced in some way. We offer sufficient conditions for optimal seeded splits; however, the challenge of discovering optimal seed budgets remains an important open combinatorial challenge. Finally, the experiments used in this paper assess canonical topologies. Stress tests on the algorithm, such as those on dynamically varying networks or with real social data, would be helpful for future study in these areas.

\section{Conclusion}
This work shows that the one-step threshold for Heaven-Hell consensus scales robustly to tolerances, tie policies, seeding, multiple hubs, and asynchronous schedules. A simple conservation-law perspective yields both exact thresholds and sharp, interpretable bounds. Experiments demonstrate their tightness and practical utility.

\paragraph{Artefact.} All results are machine-checked in Coq. The accompanying source code includes the formal proofs (e.g., \texttt{ScaleLaw}, \texttt{Monotonicity}) and Python scripts (\texttt{showcase.py}) to reproduce all experiments.

\newpage
\bibliographystyle{unsrt}
\bibliography{ref}
\newpage

\appendix
\section{Proofs of Supporting Lemmas}

This appendix provides formal proofs for technical lemmas referenced in the main paper.

\begin{lemma}[Equivalence of Per-Node and Aggregated Thresholds]
\label{lem:max_need_equiv_appendix}
Let $W$ be the uniform hub weight. The per-node condition $(\forall v \neq g, W + \tau(v) \ge \mathrm{rest\_weight}(v))$ is equivalent to the aggregated condition $W \ge \max_{v \neq g}(\mathrm{rest\_weight}(v) - \tau(v))$.
\end{lemma}
\begin{proof}
\begin{itemize}
    \item[($\Rightarrow$)] Assume $\forall v \neq g, W + \tau(v) \ge \mathrm{rest\_weight}(v)$. This is equivalent to $\forall v \neq g, W \ge \mathrm{rest\_weight}(v) - \tau(v)$. Since $W$ is greater than or equal to every element in the set $\{\mathrm{rest\_weight}(v) - \tau(v)\}_{v \neq g}$, it must be greater than or equal to the maximum element of that set. Thus, $W \ge \max_{v \neq g}(\mathrm{rest\_weight}(v) - \tau(v))$.
    \item[($\Leftarrow$)] Assume $W \ge \max_{v \neq g}(\mathrm{rest\_weight}(v) - \tau(v))$. By the definition of maximum, for any specific $v_0 \neq g$, we have $\max_{v \neq g}(\mathrm{rest\_weight}(v) - \tau(v)) \ge \mathrm{rest\_weight}(v_0) - \tau(v_0)$. By transitivity, $W \ge \mathrm{rest\_weight}(v_0) - \tau(v_0)$, which rearranges to $W + \tau(v_0) \ge \mathrm{rest\_weight}(v_0)$. Since this holds for any arbitrary $v_0$, the per-node condition is satisfied.
\end{itemize}
\end{proof}

\begin{lemma}[State Monotonicity of Scores]
\label{lem:state_monotonicity_appendix}
Let states $s, t$ be such that for all $u$, if $s(u)=\textsc{Glory}$ then $t(u)=\textsc{Glory}$. Then for any node $v$, $\mathrm{score}_G(s,v) \le \mathrm{score}_G(t,v)$ and $\mathrm{score}_N(s,v) \ge \mathrm{score}_N(t,v)$.
\end{lemma}
\begin{proof}
Let $s'$ and $t'$ be the states after forcing the hub to \textsc{Glory}. The premise implies that the set of \textsc{Glory} nodes under $s'$ is a subset of those under $t'$. The score $\mathrm{score}_G(s,v) = \sum_{u} w(u,v) \mathbb{I}[s'(u) = \textsc{Glory}]$ is a sum of non-negative terms. Since the sum for $\mathrm{score}_G(t,v)$ is over a superset of terms (or the same set), we have $\mathrm{score}_G(s,v) \le \mathrm{score}_G(t,v)$. The second inequality, $\mathrm{score}_N(s,v) \ge \mathrm{score}_N(t,v)$, follows directly from this and the conservation law (Lemma \ref{lem:score_conservation}): $\mathrm{total\_in}(v) - \mathrm{score}_G(s,v) \ge \mathrm{total\_in}(v) - \mathrm{score}_G(t,v)$.
\end{proof}

\begin{lemma}[Scores under All-\textsc{Gnash} Initialization]
\label{lem:all_gnash_scores_appendix}
Let $H$ be a seed set. In an initial state where $s(u) = \textsc{Gnash}$ for all $u \notin H$, the scores for a node $v \notin H$ after forcing $H$ to \textsc{Glory} are exactly $\mathrm{score}_G(s,v) = \mathrm{hub\_H}(v)$ and $\mathrm{score}_N(s,v) = \mathrm{rest\_H}(v)$.
\end{lemma}
\begin{proof}
After forcing nodes in $H$ to \textsc{Glory}, the state $s'$ is such that $s'(u) = \textsc{Glory}$ if $u \in H$ and $s'(u) = \textsc{Gnash}$ if $u \notin H$. By definition, the scores for node $v$ are computed from this state $s'$.
\begin{align*}
\mathrm{score}_G(s,v) &= \sum_{u \in V} w(u,v) \mathbb{I}[s'(u) = \textsc{Glory}] = \sum_{u \in H} w(u,v) = \mathrm{hub\_H}(v). \\
\mathrm{score}_N(s,v) &= \sum_{u \in V} w(u,v) \mathbb{I}[s'(u) = \textsc{Gnash}] = \sum_{u \notin H} w(u,v) = \mathrm{rest\_H}(v). \qedhere
\end{align*}
\end{proof}

\end{document}